\documentclass{llncs}

\usepackage{amsmath}
\usepackage{amssymb}
\usepackage{mathrsfs}
\usepackage{graphicx}
\usepackage{enumitem}
\usepackage{mathtools}
\usepackage{setspace}
\newtheorem{rem}[theorem]{Remark}

\providecommand{\openbox}{\leavevmode
  \hbox to.77778em{%
  \hfil\vrule
  \vbox to.675em{\hrule width.6em\vfil\hrule}%
  \vrule\hfil}}
\makeatletter
\DeclareRobustCommand{\qed}{%
  \ifmmode
    \eqno \def\@badmath{$$}
    \let\eqno\relax \let\leqno\relax \let\veqno\relax
    \hbox{\openbox}%
  \else
    \leavevmode\unskip\penalty9999 \hbox{}\nobreak\hfill
    \quad\hbox{\openbox}%
  \fi
}
\makeatother

\renewcommand{\geq}{\geqslant}
\renewcommand{\leq}{\leqslant}
\newcommand{\pow}{\mathscr{P}}

\newcommand{\sT}{\mid}

\newcommand{\N}{{\nats_A}}

\newcommand{\R}{{\mathbb{R}_A}}

\newcommand{\Rs}{{\mathbb{R}^\mu _A}}

\newcommand{\rk}{{\mathsf{rk}}}

\newcommand{\TrCl}[1]{{\mathsf{trCl}\left(#1\right)}}

\newcommand{\absb}[1]{\big|#1\big|}
\newcommand{\abs}[1]{\left|#1\right|}

\newcommand{\defAs}{\coloneqq}

\newcommand{\nats}{\mathbb{N}}

\newcommand{\HF}{\mbox{$\mathsf{HF}$}}

\newcommand{\HHF}{\mbox{$\HF^{1/2}$}}

\newcommand{\reals}{\mathbb{R}}

\newcommand{\om}{\mathbf{\Omega}}

\newcommand{\csy}{\mathscr{C}}
\newcommand{\ses}{\mathscr{S}}

\newcommand{\Rel}{\mathrel{\mathsf{R}}}

\graphicspath{{./figures/}} 

\newcommand{\x}{\varsigma}
\newcommand{\xsi}{x}

\title{Encoding Sets as Real Numbers\\[.1cm] \large{(Extended version)}}

\author{Domenico Cantone\inst{1} \and Alberto Policriti\inst{2}}

\institute{Dept.\ of Mathematics and Computer Science, University of Catania, Italy \and Dept.\ of Mathematics, Computer Science, and Physics, University of Udine, Italy}

\date{\today}

\begin{document}

	    
 \maketitle
 
\begin{abstract}
We study a variant of the Ackermann encoding $\N(x)  \defAs  \sum_{y\in x}2^{\N(y)}$ of the hereditarily finite sets by the natural numbers, applicable to the larger collection $\HHF$ of the hereditarily finite hypersets. The proposed variation is obtained by simply placing a `minus' sign before each exponent in the definition of $\N$, resulting in the expression $\R(x) \defAs  \sum_{y\in x}2^{-\R(y)}$. By a careful analysis, we prove that the encoding $\R$ is well-defined over the whole collection $\HHF$, as it allows one to univocally assign a real-valued code to each hereditarily finite hyperset. 
We also address some preliminary cases of the injectivity problem for $\R$.
\end{abstract}

 \section{Introduction}
 In 1937, W.\ Ackermann proposed the following recursive encoding of hereditarily finite sets by natural numbers:
  \begin{equation}
    \N(x)  \defAs  \sum_{y\in x}2^{\N(y)}.\label{NA}
  \end{equation}
 The encoding $\N$ is simple, elegant, and highly expressive for a number of reasons. On the one hand, it builds a strong bridge between two foundational mathematical structures: (hereditarily finite) sets and 
(natural) numbers. 
On the other hand, it enables the representation of the \emph{characteristic function} of hereditarily finite sets in terms of the usual notation for natural numbers as sequences of binary digits. That is: $y$ belongs to $x$ if and only if the $\N(y)$-th digit in the binary expansion of $\N(x)$ is equal to 1. As one would expect, the string of 0's and 1's representing $\N(x)$ is nothing but (a representation of) the characteristic function of $x$.
 
In this paper we study a very simple variation of the encoding $\N$, originally proposed in \cite{DBLP:conf/cilc/Policriti13a} and discussed in \cite{DBLP:journals/fuin/DAgostinoOPT15} and \cite{OmodeoPT2017},  applicable to a larger collection of sets. The proposed variation is obtained by simply placing a minus sign before each exponent in \eqref{NA}, resulting in the expression:
\[ 
    \R(x) \defAs  \sum_{y\in x}2^{-\R(y)}.
\]
As opposed to the encoding $\N$, the range of $\R$ is not contained in the set $\nats$ of the natural numbers, but it extends to the set $\reals$ of real numbers. In addition, the domain of $\R$ can be expanded so as to include also the \emph{non-well-founded} hereditarily finite sets, namely, the sets defined by (finite) systems of equations of the following form  
\begin{equation}\label{system}
\left\{
	\begin{aligned}
		\x_1 & =   \{\x_{1,1}, \ldots , \x_{1,m_1}\} \\
		 &   \hspace{.2cm}\vdots \\
		\x_n & =   \{\x_{n,1}, \ldots , \x_{n,m_n}\}, \\
	\end{aligned}
\right.
\end{equation}
with \emph{bisimilarity} as equality criterion (see \cite{Acz} and \cite{BM}, where the term \emph{hyperset} is also used). For instance, the special case of the single set equation $\x=\{\x\}$, resulting into the equation (in real numbers) $\xsi=2^{-\xsi}$, is illustrated in Section~\ref{new} and
provides the code of the unique (under bisimilarity) hyperset $\Omega=\{\Omega\}$.\footnote{Notice that the solution to the equation $\xsi=e^{-\xsi}$ is the so-called \emph{omega} constant, introduced by Lambert in \cite{Lambert1758} and studied also by Euler in \cite{EulerOmnia}.}

While the encoding $\N$ is defined inductively (and this is perfectly in line with our intuition of the very basic properties of the collections of natural numbers $\nats$ and of hereditarily finite sets $\HF$---called $\HF^0$ in \cite{BM}), the definition of $\R$, instead, is not inductive when extended to non-well-founded sets, and thus it requires a more careful analysis, as it must be \emph{proved} that it univocally (and possibly injectively) associates (real) numbers to sets. 
 
The injectivity of $\R$ on the collection of well-founded and non-well-founded hereditarily finite sets---henceforth, to be referred to as $\HHF$, see  \cite{BM}---was conjectured in \cite{DBLP:conf/cilc/Policriti13a} and is still an open problem. Here we prove that, given any finite collection $\hbar_1, \ldots, \hbar_n$ of pairwise distinct sets in $\HHF$ satisfying a system of set-theoretic equations of the form \eqref{system} in $n$ unknowns, one can univocally determine real numbers $\R(\hbar_1) , \ldots , \R(\hbar_n)$ satisfying the following system of equations:
\[\left\{
\begin{aligned}
 \R(\hbar_1) & =  \textstyle\sum_{k = 1}^{m_1}2^{-\R(\hbar_{1,k})}\\
   &\hspace{.2cm}\vdots\\
 \R(\hbar_n) & =  \textstyle\sum_{k = 1}^{m_n}2^{-\R(\hbar_{n,k})}.
\end{aligned}
\right. \]

This preliminary result shows that the definition of $\R$ is well-given, as it associates a unique (real) number to each hereditarily finite hyperset in $\HHF$. This extends to $\HHF$ the first of the properties that the encoding $\N$ enjoys with respect to $\HF$. Should $\R$ also enjoy the injectivity property, the proposed adaptation of $\N$ would be completely satisfiying, and $\R$ could be coherently dubbed an \emph{encoding} for $\HHF$. 
 
In the course of our proof, we shall also present a procedure that drives us to the real numbers $\R(\hbar_1) , \ldots , \R(\hbar_n)$ mentioned above by way of successive approximations. In the well-founded case, our procedure will converge in a finite number of steps, whereas infinitely many steps will be required for convergence in the non-well-founded case. In the last section of the paper, we shall also briefly hint at the injectivity problem for $\R$.
 
 \section{Basics}\label{basics}
Let $\nats$ be the set of natural numbers and let $\pow(\cdot)$ denote the \emph{powerset} operator.

\begin{definition}[Hereditarily finite sets] $\HF \defAs \bigcup_{n \in \nats} \HF_n$ is the collection of all \emph{hereditarily finite} sets, where
\[
\left\{
\begin{aligned}
\HF_0 &\defAs \emptyset\/,\\
\HF_{n+1} &\defAs \pow(\HF_n)\/, \quad \text{for } n \in \nats.
\end{aligned}
\right.
\]
\end{definition}

In this paper we shall introduce and study a variation of  the following map introduced by Ackermann in 1937 (see \cite{Ackermann-37}): 
\begin{definition}[Ackermann encoding]\label{Ack}
  \begin{align*} 
    \N(h) & \defAs  \sum_{h'\in h}2^{\N(h')}\/, \quad \text{for } h \in \HF.
  \end{align*}
\end{definition}
It is easy to see that the  map $\N$ is a bijection between $\HF$.

From now on, we denote by $h_i$ the element of $\HF$ whose $\N$-code is $i$, so that $\N(h_{i}) = i$, for $i \in \nats$.
Using the iterated-singleton notation 
\[
\begin{aligned}
\{x\}^{0} &\defAs x\\
\{x\}^{n+1} &\defAs \big\{\{x\}^{n}\big\}, && \text{for $n \in \nats$,}
\end{aligned}
\]
we plainly have:
\[
h_{0} = \emptyset,\quad h_{1} = \{\emptyset\},\quad h_{2} = \{\emptyset\}^{2},\quad h_{3} = \{\{\emptyset\},\emptyset\},\quad h_{4} = \{\emptyset\}^{3},\quad \text{etc.}
\]
In addition, for $j \in \nats$ we have:
\begin{gather}
h_{0} \in h_{j}\quad \text{iff} \quad j \text{ is odd,}\label{jOdd}\\
h_{2^{j}} = \{h_{j}\} \quad \text{and} \quad 
h_{2^{j}-1} = \{h_{0},h_{1},\ldots,h_{j-1}\}\,.\label{2toj}
\end{gather}
The  map $\N$ induces a total ordering $\prec$ among the elements of $\HF$ (that we shall call \emph{Ackermann ordering}) such that, for $h,h' \in \HF$:
\[
h \prec h' \qquad \text{iff} \qquad \N(h) < \N(h')\,.
\]
Thus, $h_{i}$ is the $i$-th element of $\HF$ in the Ackermann ordering and, for $i,j \in \nats$, we plainly have:
\[
h_i \prec h_j \qquad \text{iff} \qquad i < j.
\]

The following proposition can be read as a restating of the bitwise comparison between natural numbers in set-theoretic terms.
\begin{proposition}\label{antilex}
For $h_i, h_j \in \HF$, the following equivalence holds: 
\[
h_i \prec h_j \qquad \text{iff} \qquad h_i \neq h_j \text{~and~} \max_\prec(h_i \mathrel{\Delta} h_j) \in h_{j}\,,
\]
where $\Delta$ is the symmetric difference operator $A \mathrel{\Delta} B \defAs (A \setminus B) \cup (B \setminus A)$.
\end{proposition}
\begin{proof}
Since $h_i \prec h_j$ is equivalent to $i < j$, it is sufficient to
 compare the base two expansions of $i$ and $j$ and apply Definition \ref{Ack}. \qed
\end{proof}

\newcommand{\low}{\mathsf{low}}
It is also useful to define the following map $\low \colon \nats \rightarrow \nats$
which, for $i \in \nats$, computes the smallest code $j$ of a set $h_{j}$ not present in $h_{i}$ (or, equivalently, the position---starting from $0$---of the lowest bit set to $0$ in the binary expansion of $i$): 
\[
\low(i) \defAs \min \{j \sT h_{j} \notin h_{i}\}\,.
\]

The following elementary properties, whose proof is left to the reader, will be used in the last section of the paper.
\begin{lemma}\label{lowLemma}
For $i \in \nats$, we have:
\begin{enumerate}[label=(\roman*),leftmargin=.8cm]
\item\label{lowLemmai} $\low(i) = 0$ \quad iff \quad  $i$ is even,

\item\label{lowLemmaii} $h_{\low(i)} \notin h_{i}$,

\item\label{lowLemmaiii} $\{h_{0},\ldots,h_{\low(i)-1}\} \subseteq h_{i}$,

\item\label{lowLemmaiv} $h_{i+1} = \big( h_{i} \:\setminus\: \{h_{0},\ldots,h_{\low(i)-1}\} \big) \cup \{h_{\low(i)}\}$.
\end{enumerate}
\end{lemma}

Finally, we briefly recall that hypersets satisfy all axioms of ZFC, but the axiom of regularity, which forbids both membership cycles and infinite descending chains of memberships. In the hypersets realm of our interest, based on the Forti-Honsell axiomatization \cite{FH83} as popularized by P.~Aczel \cite{Acz}, the axiom of regularity is replaced by the anti-foundation axiom (AFA). 
Roughly speaking, AFA states that every conceivable hyperset, described in terms of a graph specification modelling  its internal membership structure, actually exists and is univocally determined, regardless of the presence of cycles or infinite descending paths in its graph specification. 
To be slightly more precise, a graph specification (or membership graph) is a directed graph with a distinguished node (\emph{point}), where nodes are intended to represent hypersets, edges model membership relationships among the node/hypersets, and the distinguished node denotes the hyperset of interest among all the hypersets represented by the nodes of the graph.
However, extensionality needs to be strengthened so as structurally indistinguishable pointed graphs are always realized by the same hyperset.
More specifically, we say that two pointed graphs $G = (V,E,p)$ and $G' = (V',E',p')$, where $p \in V$ and $p' \in V'$ are the `points' of $G$ and $G'$, respectively, are \emph{structurally indistinguishable} if the points $p$ and $p'$ are \emph{bisimilar}, namely, if there exists a relation $\Rel$ over $V \times V'$ such that the following three properties hold, for all $v\in V$ and $v'\in V'$:
\begin{enumerate}[label=(B\arabic*),leftmargin=.9cm]
\item\label{B1} $(\forall w \in V)(\exists w' \in V') \big( (v \Rel v' \wedge (v,w) \in E) \rightarrow ((v',w') \in E' \wedge w \Rel w')\big)$;

\item\label{B2} $(\forall w' \in V')(\exists w \in V) \big( (v \Rel v' \wedge (v',w') \in E') \rightarrow ((v,w) \in E \wedge w \Rel w')\big)$;

\item $p \Rel p'$.
\end{enumerate}
Any relation $\Rel' \subseteq V \times V'$ satisfying properties \ref{B1} and \ref{B2} above is called a \emph{bisimulation} over $V \times V'$.

\smallskip

In this paper, we are interested in those hypersets that admit a representation as a \emph{finite} pointed graph: these are the \emph{hereditarily finite hypersets} in $\HHF$.

Given a (finite) pointed graph $G=(V,E,p)$ whose nodes are all reachable from its point $p$, the collection of hypersets represented by all its nodes form the \emph{transitive closure} of $\{\hbar\}$,  denoted $ \TrCl{\{\hbar\}} $, where $\hbar$ is the hyperset corresponding to the point $p$.





%

\section{The real-valued map $\R$}\label{new}
Consider the following map $\R$ obtained from  $\N$ by simply placing a minus sign before each exponent in (\ref{NA}):
  \begin{align}
    \R(x)  &\defAs  \sum_{y\in x} 2^{-\R(y)}.\label{RA}
  \end{align}

From \eqref{RA} it follows immediately that all (valid) $\R$-codes are nonnegative. For instance, we have:
\[
\begin{aligned}
\R(\emptyset) &= 0, & \R(\{\emptyset\}) &= 1, & \R(\{\emptyset\}^{2}) &= \frac{1}{2}, \\
\R(\{\emptyset\}^{3}) &= \frac{1}{\sqrt{2}},~~~ & \R(\{\emptyset\}^{4}) &= 2^{-\frac{1}{\sqrt{2}}}, ~~~ & \R(\{\emptyset\}^{5}) &= 2^{-2^{-\frac{1}{\sqrt{2}}}}, \quad \text{etc.} 
\end{aligned}
\]

The definition of $\R$ bears a strong formal similarity with $\N$, but calls into play real numbers. As a further example, from the definition of $\R$, it follows that the hyperset defined by the set equation $\x = \{\x\}$ yields the equation in $\reals$
\begin{align}\label{omega}
\xsi & =  2^{-\xsi}\,.
\end{align}
It is easy to see that the equation \eqref{omega} has a unique solution in $\reals$, since the functions $\xsi$ and $2^{-\xsi}$ are, respectively, strictly increasing and strictly decreasing, so that the function $\xsi-2^{-\xsi}$ is strictly increasing. In addition, we have:
\[
\xsi-2^{-\xsi}|_{\xsi = \frac{1}{2}} = \frac{1}{2} - \frac{1}{\sqrt{2}} < 0 < 1 - \frac{1}{2} = \xsi-2^{-\xsi}|_{\xsi = 1}.
\]
Thus, the solution $\om$ of (\ref{omega}) over $\reals$, namely, the code of the hyperset defined by the set equation $x = \{x\}$, satisfies $\frac{1}{2} < \om < 1$. Furthermore, much by the same argument used by the Pythagoreans to prove the irrationality of $\sqrt{2}$, it can easily be shown that $\om$ is irrational. In fact, $\om$ is transcendental. This follows from the Gelfond-Schneider theorem (see \cite{G34}), which states that every real number of the form $a^{b}$ is transcendental, provided that $a$ and $b$ are algebraic numbers such that $0 \neq a \neq 1$, and $b$ is irrational.\footnote{We recall that the Gelfond-Schneider theorem, obtained independently in 1934 by A.~O.\ Gelfond and Th.\ Schneider, solves completely the seventh in a well-celebrated list of twenty-three problems posed by David Hilbert at the \emph{International Congress of Mathematicians} held in Paris, 1900 (see~\cite{Hilbert-1902}).} Indeed, if $\om$ were algebraic, so would be $-\om$ and therefore, by the Gelfond-Schneider theorem, $2^{-\om} = \om$ would be transcendental, contradicting the assumed algebraicity of $\om$. Thus, $\om$ must be transcendental after all. As is easy to check, the $\R$-code $2^{-1/\sqrt{2}}$ of $\{\emptyset\}^4$ is transcendental as well.

Much as for $\N$, the encoding $\R$ is easily seen to be well-defined over $\HF$. As a consequence of the results to be proved in Section~\ref{hyper}, we shall see that \eqref{RA} allows one to associate univocally a code to each non-well-founded hereditarily finite set as well, thus showing that $\R$ is also well-defined over the whole collection $\HHF$ of hereditarily finite hypersets.

\begin{rem}
For every singleton $\{h'\} \in \HF$, definition \eqref{RA} gives $\R(\{h'\})  =  2^{-\R(h')}$. Thus, for every $h \in \HF$, we have 
\begin{equation}
\label{repeatedly}
\R(h) = \sum_{h' \in h}^n 2^{-\R(h')} = \sum_{h' \in h}^n \R\left(\{h'\}\right).
\end{equation}

Once we prove that $\R$ is well-defined over $\HHF$,  equation \eqref{repeatedly} can be immediately generalized to $\HHF$ as well.\qed
\end{rem}

As the following proposition shows, the codes of hereditarily finite sets can grow arbitrarily large.

\begin{proposition}
For any $n\in \nats$, there exists an $i \in \nats$ such that $\R(h_i)>n$.
\end{proposition}
\begin{proof}
Notice that for any odd natural number $j$, we have $\emptyset \in h_{j}$. Thus, by \eqref{repeatedly}, we have $\R(h_{j}) \geq R(\{\emptyset\}) = 1$, $\R(\{h_{j}\}) =2^{-\R(h_{j})} \leq 2^{-1} = \frac{1}{2}$, and $\R(\{\{h_{j}\}\}) =2^{-\R(\{h_{j}\})} \geq 2^{-1/2} > \frac{1}{2}$.

Let $k=4n$ and consider the hereditarily finite set $h \defAs \big\{\{h_{k'}\} \ : \  k'\leq k\big\}$. Then, we have:
\[
\R(h) =\sum_{k'=0}^{k} \R\left(\left\{\{h_{k'}\}\right\}\right) \geq \sum_{\substack{k'=0 \\ k' \text{ is odd}}}^{k} \R\left(\left\{\{h_{k'}\}\right\}\right) > \frac{1}{2} \cdot \frac{k}{2} = n\,. \qed
\]
\end{proof}

\section{$\R$ on Hereditarily Finite Hypersets}\label{hyper}

The fully general case  corresponds to considering systems 
of set-theoretic equations such as the ones introduced by the following definition (see also \cite{Acz}).
\newcommand{\idxmp}{I_{\!\ses}}
\begin{definition}[\textbf{Set systems}] 
	A \emph{set system $\ses(\x_1, \ldots, \x_n)$} in the set unknowns $\x_1, \ldots , \x_n$ is a collection of set-theoretic equations of the form:
\begin{equation}\label{setSystem}
\left\{
	\begin{aligned}
		\x_1 & =   \{\x_{1,1}, \ldots , \x_{1,m_1}\} \\
		 &   \hspace{.2cm}\vdots \\
		\x_n & =   \{\x_{n,1}, \ldots , \x_{n,m_n}\}, \\
	\end{aligned}
\right.
\end{equation}
with $m_i \geq 0$ for $i\in \{1, \ldots, n\}$, and where each unknown $\x_{i,u}$, for $i\in \{1, \ldots, n\}$ and $u \in \{ 1, \ldots, m_i\}$, occurs among the unknowns $\x_1, \ldots , \x_n$.\footnote{When $m_i = 0$, the expression $\{\x_{i,1}, \ldots , \x_{i,m_i}\}$ reduces to the empty set  expression $\{\}$.}
	
The \emph{index map} $\idxmp$ of $\ses(\x_1, \ldots, \x_n)$ is the map 
\[
\idxmp \colon \bigcup_{i=1}^{n} \{\langle i,v\rangle \sT 1 \leq v \leq m_{i}\} \rightarrow \{1,\ldots,n\}
\]
such that $\idxmp(i,v)$ is the index of the unknown $\x_{i,v}$ in the list $\x_1, \ldots, \x_n$, for $i \in \{1,\ldots,n\}$ and $v \in \{1,\ldots,m_{i}\}$, namely, $\x_{\idxmp(i,v)} \equiv \x_{i,v}$.

A set system $\ses(\x_1, \ldots, \x_n)$ is \emph{normal} if there exist $n$ pairwise distinct (i.e., non bisimilar) hypersets $\hbar_1, \ldots, \hbar_n \in \HHF$ such that the assignment $\x_i \mapsto \hbar_i$ satisfies all the set equations of $\ses(\x_1, \ldots, \x_n)$.
\end{definition}
Observe that, by the anti-foundation axiom, the assignment $\x_i \mapsto \hbar_i$ satisfying the equations of a given normal set system $\ses(\x_1, \ldots, \x_n)$ is plainly unique.

From now on, we shall write $\hbar$, with or without subscripts and/or superscripts, to denote a generic (possibly well-founded) hyperset in $\HHF$.

\smallskip

Having in mind the definition \eqref{RA} of $\R$, to each normal set system we associate a system of equations in $\reals$, called \emph{code system}, as follows.

\begin{definition}[Code systems]
Let $\ses(\x_1, \ldots, \x_n)$ be a normal set system of the form 
\begin{equation*}
\left\{
	\begin{aligned}
		\x_1 & =   \{\x_{1,1}, \ldots , \x_{1,m_1}\} \\
		 &   \hspace{.2cm}\vdots \\
		\x_n & =   \{\x_{n,1}, \ldots , \x_{n,m_n}\}, \\
	\end{aligned}
\right.
\end{equation*}
with index map $\idxmp$. The \emph{code system} associated with $\ses(\x_1, \ldots, \x_n)$ is the following system $\csy(\xsi_{1},\ldots,\xsi_{n})$ of equations in the real unknowns $\xsi_{1},\ldots,\xsi_{n}$:
\begin{equation}\label{codeSystem}
\left\{
	\begin{aligned}
	\xsi_1 & =   2^{-\xsi_{1,1}} + \cdots + 2^{-\xsi_{1,m_1}} \\
		 &   \hspace{.2cm}\vdots \\
	\xsi_n & =   2^{-\xsi_{n,1}} + \cdots + 2^{-\xsi_{n,m_n}}, \\
	\end{aligned}
	\right. 
\end{equation}
where $\xsi_{i,u}$ is a shorthand for $\xsi_{\idxmp(i,u)}$, for $i \in \{1,\ldots,n\}$ and $u \in \{1,\ldots,m_{i}\}$.
\end{definition}

\medskip

Normal set systems characterise all possible elements of $\HHF$ and we shall see that the corresponding code systems characterise their $\R$-codes. In fact, we shall prove that every code system admits a unique solution which can be computed as the limit of suitable sequences of real numbers. Terms in such sequences approximate the final solution alternatively from above and from below. In case of non-well-founded sets, such approximating sequences are infinite and convergent (to the codes of the non-well-founded sets); additionally, its terms eventually become codes of certain well-founded hereditarily finite sets which can be seen as approximations of the corresponding non-well-founded set.

\smallskip

We begin by formally defining the \textit{set} and \textit{multi-set approximating sequences} (of the solution) of set systems. 

\begin{definition}[Hereditarily finite multi-sets]\label{multi-sets}
\emph{Hereditarily finite multi-sets} are collection of elements---themselves multi-sets---that can occur with finite multiplicities.
\end{definition}

Hereditarily finite multi-sets will be denoted by specifying their elements in square brackets, in any order, where elements are repeated according to their multiplicities. 
For instance, the set $ a $ occurs in the multi-set $ [b, a, b, b] $ with multiplicity 1, whereas $ b $ occurs with multiplicity 3.

\begin{rem}\label{multi-set-set}
	A natural embedding of the hereditarily finite sets in the hereditarily finite multi-sets is the following: a multi-set $ \mu $ can  be regarded as a set if and only if each of its elements has multiplicity 1 and, recursively, can be regarded as a set. Thus, in particular, $ \emptyset $ is both a set and a multi-set.\qed
\end{rem} 


\begin{definition} Let $\ses(\x_1, \ldots, \x_n)$ be a (normal) set system of the form \eqref{setSystem}, and let $\idxmp$ be its index map.
	 The \emph{set-approximating sequence} for (the solution of) $\ses(\x_1, \ldots, \x_n)$ is the sequence $\big\{\langle\hbar_{i}^{j} \sT 1 \leq i \leq n\rangle\big\}_{j \in \nats}$ of the $n$-tuples of well-founded hereditarily finite sets defined by
	\begin{align*}
	\langle \hbar_i^{j} \sT 1 \leq i \leq n\rangle \defAs & \left\{\begin{array}{ll}
	 \langle\emptyset  \sT 1 \leq i \leq n\rangle & \mbox{ if } j=0  \\
	 & \\
	 \langle\{\hbar^{j-1}_{i,1}, \ldots , \hbar^{j-1}_{i,m_i}\} \sT 1 \leq i \leq n\rangle & \mbox{ if } j > 0\,, 
	\end{array}
	\right.
	\end{align*}
where $\hbar^{j-1}_{i,u}$ is a shorthand for $\hbar^{j-1}_{\idxmp(i,u)}$, for $i \in \{1,\ldots,n\}$ and $u \in \{1,\ldots,m_{i}\}$.

	 The \emph{multi-set approximating sequence} for (the solution of) $\ses(\x_1, \ldots, \x_n)$ is the sequence $\big\{\langle \mu_i^{j} \sT 1 \leq i \leq n\}\rangle\big\}_{j \in \nats}$ of the $n$-tuples of well-founded hereditarily finite multi-sets defined by
	 \begin{align*}
		 	\langle \mu_i^{j} \sT 1 \leq i \leq n\rangle \defAs & \left\{\begin{array}{ll}
		 		\langle\emptyset  \sT 1 \leq i \leq n\rangle & \mbox{ if } j=0  \\
		 		& \\
		 		\langle[ \mu ^{j-1}_{i,1}, \ldots , \mu ^{j-1}_{i,m_i}] \sT 1 \leq i \leq n\rangle & \mbox{ if } j >0\,,
		 	\end{array}
		 	\right.
	 \end{align*}
where $\mu^{j-1}_{i,u}$ is a shorthand for $\mu^{j-1}_{\idxmp(i,u)}$, for $i \in \{1,\ldots,n\}$ and $u \in \{1,\ldots,m_{i}\}$.		 
\end{definition}

Given a (normal) set system $\ses(\x_1, \ldots, \x_n)$, we say that two unknowns $\x_{i}$ and $\x_{i'}$, with $i,i' \in \{1,\ldots,n\}$, are \emph{distinguished} at step $k \geq 0$ by the set-approximating sequence $\big\{\langle\hbar_{i}^{j} \sT 1 \leq i \leq n\rangle\big\}_{j \in \nats}$ for $\ses(\x_1, \ldots, \x_n)$ (resp., multi-set approximating sequence $\big\{\langle\mu_{i}^{j} \sT 1 \leq i \leq n\rangle\big\}_{j \in \nats}$ for $\ses(\x_1, \ldots, \x_n)$) if $\hbar_{i}^{k} \neq \hbar_{i'}^{k}$ (resp., $\mu_{i}^{k} \neq \mu_{i'}^{k}$). Further, we shall refer to $\hbar_{i}^{j}$ (resp., $\mu_{i}^{j}$) as the \emph{($j$-th) set-approximation value} (resp., \emph{multi-set approximation value}) of $\x_{i}$ at step $j$.

\begin{example}
	Consider the following normal set system:
	\[
	\ses(\x_1,\x_2,\x_3,\x_4) = 
	\left\{
	\begin{array}{ccl}
	\x_1 & = &  \{\x_{2},  \x_{3}\} \\
	\x_2 & = & \{\} \\
	\x_3 & = &  \{\x_{3}\} \\
	\x_4 & = &  \{\x_{2}\}\,. \\
	\end{array}
	\right. \]
In this case, we have: $m_1=2, m_2=0, m_3=1$, and $m_4=1$; in addition, $\x_{1,1}$, $\x_{1,2}$, $\x_{3,1}$, and $\x_{4,1}$ are $\x_2$, $\x_3$,  $\x_3$, and $\x_{2}$, respectively, so that $\idxmp(1,1) = 2$, $\idxmp(1,2) = 3$, $\idxmp(3,1) = 3$, and $\idxmp(4,1) = 2$.
	
The first four elements of the set-approximating sequence for $ \ses $ are: 
\begin{align*}
		\left\langle \hbar^0_1, \hbar^0_2,\hbar^0_3, \hbar^0_4 \right\rangle = & 
		\left\langle \emptyset,\emptyset,\emptyset,\emptyset  \right\rangle \\
		\left\langle \hbar^1_1, \hbar^1_2,\hbar^1_3, \hbar^1_4 \right\rangle = & \left\langle \{\emptyset\},\emptyset,\{\emptyset\},\{\emptyset\}  \right\rangle \\
		\left\langle \hbar^2_1, \hbar^2_2,\hbar^2_3, \hbar^2_4 \right\rangle = & \left\langle \{\emptyset, \{\emptyset\}\},\emptyset,\{\{\emptyset\}\},\{\emptyset\}  \right\rangle \\
		\left\langle \hbar^3_1, \hbar^3_2,\hbar^3_3, \hbar^3_4 \right\rangle = & \left\langle \{\emptyset, \{\{\emptyset\}\}\},\emptyset,\{\{\{\emptyset\}\}\},\{\emptyset\}  \right\rangle \,.
\end{align*}
	
	Notice that, for $ j=2 $ and $j=3$, all the $ \hbar^j_i $'s are pairwise distinct. In fact, as a consequence of the next lemma, this is true also for all $ j>3 $.
	
	The first four tuples of the  multi-set approximating sequence of $ \ses $ are: 
	\begin{align*}
		\left\langle \mu ^0_1, \mu ^0_2,\mu ^0_3, \mu ^0_4 \right\rangle = & \left\langle \emptyset,\emptyset,\emptyset,\emptyset  \right\rangle \\
		\left\langle \mu ^1_1, \mu ^1_2,\mu ^1_3, \mu ^1_4 \right\rangle = & \left\langle [\emptyset,\emptyset],\emptyset,[\emptyset],[\emptyset]  \right\rangle \\
		\left\langle \mu ^2_1, \mu ^2_2,\mu ^2_3, \mu ^2_4 \right\rangle = & 
		\left\langle [\emptyset,[\emptyset]],\emptyset,[[\emptyset]],[\emptyset]  \right\rangle \\
		\left\langle \mu ^3_1, \mu ^3_2,\mu ^3_3, \mu ^3_4 \right\rangle = & 
		\left\langle [\emptyset, [[\emptyset]]],\emptyset,[[[\emptyset]]],[\emptyset]  \right\rangle \,.
	\end{align*}
	
	Also in this case, for $j=2$ and $ j=3 $, all the $ \mu ^j_i $'s are pairwise distinct and this holds also for $ j>3 $. Observe that the unknowns $\x_{i}$ and $\x_{j}$ are distinguished by the multi-set approximating sequence before than by the set-approximating sequence: indeed, $\mu_{1}^{1} \neq \mu_{1}^{3}$, whereas $\hbar_{1}^{1} = \hbar_{1}^{3}$. \qed
\end{example}

All sets in the tuples of a set-approximating system are well-founded hereditarily finite sets. The initial tuples of a multi-set approximating sequence may contain \textit{proper} multi-sets (as in the above example). However, as a consequence of the following lemma, after at most $n$ steps all pairs of distinct unknowns are distinguished and each tuple contains only proper sets.
 
\begin{lemma}\label{sempredistinti}
	Let $\ses(\x_1, \ldots, \x_n)$ be a normal set-system, and  let  $\big\{\langle \hbar_i^{j} \sT 1 \leq i \leq n\} \rangle\big\}_{j \in \nats}$ and   	$\big\{\langle \mu _i^{j} \sT 1 \leq i \leq n\rangle\big\}_{j \in \nats}$ be its set- and multi-set approximating sequences, respectively. 
		The following conditions hold:
\begin{itemize}
		\item[(a)]  for $i,i' \in \{1,\ldots,n\}$ and $j,k \geq 0$, we have
		\[ \hbar^j_i \neq \hbar^j_{i'} \:\Longrightarrow\: (\hbar^{j+k+1}_i \neq \hbar^{j+k+1}_{i'}  \wedge \mu ^j_i \neq \mu ^j_{i'})\]
		(namely, if at step $j \geq 0$ two unknowns are distinguished by the set-ap\-prox\-i\-mat\-ing sequence, then they are also distinguished by the multi-set approximating sequence; in addition, they remain distinguished in all subsequent steps);
		
		\item[(b)] for $j,k \geq 0$ and $i\in \{1,\ldots,n\}$, we have
		\[ \hbar^j_i = \hbar^{j+1}_{i} \:\Longrightarrow\:  \hbar^{j}_i = \hbar^{j+k+2}_{i}\]
		(namely, as soon as $\hbar^j_i = \hbar^{j+1}_{i}$ holds for some $j \geq 0$, the set-approximation value of $\x_{i}$ remains unchanged for all subsequent steps);
		\item[(c)] for $i,i'\in \{1, \ldots n\}$ and $j \geq n$, we have
		\[(i \neq i' \Longrightarrow \hbar^j_i \neq \hbar^j_{i'}) \:\wedge\:  \langle \mu _i^{j} \sT 1 \leq i \leq n \rangle = \langle \hbar_i^{j} \sT 1 \leq i \leq n \rangle \]
(namely, starting from step $n$, all pairs of distinct unknowns are distinguished and the set- and multi-set approximating sequences coincide).
	\end{itemize} 
\end{lemma}

\begin{proof}
For (a), first of all we prove by induction on $j$ that if $\hbar^j_i \neq \hbar^j_{i'}$, then $\hbar^{j+k}_i \neq \hbar^{j+k}_{i'}$, for all distinct $i,i'\in \{1, \ldots ,n\}$ and every $k > 0$. 
	
The base case $j=0$ holds trivially, since  $\hbar^0_i = \emptyset = \hbar^0_{i'}$. 
	
For the inductive step, let $j > 0$ and assume for contradiction that there exist pairwise distinct $i,i'$ such that $\hbar^j_i \neq \hbar^j_{i'}$, while $\hbar^{j+k}_i = \hbar^{j+k}_{i'}$, for some $k>0$. As $\hbar^j_i \neq \hbar^j_{i'}$, we can suppose w.l.o.g.\ that $ \hbar^{j-1}_{i,u} \notin \hbar^j_{i'}$, for some $u\in \{ 1, \ldots , m_i\}$. From the fact that $\hbar^{j+k}_i = \hbar^{j+k}_{i'}$, it follows that $\hbar^{j+k-1}_{i,u} \in \hbar^{j+k}_{i'}$. Hence, we have $\hbar^{j+k-1}_{i,u} = \hbar^{j+k-1}_{i',v}$ and $\hbar^{j-1}_{i,u} \neq \hbar^{j-1}_{i',v}$, for some  $v\in \{ 1, \ldots , m_{i'}\}$, contradicting the inductive hypothesis relative to $j-1$, for $k$ and the indices $\idxmp(i,u)$ and $\idxmp(i',v)$.
		
	To complete case (a), next we prove, again by induction on $ j $, that if $\hbar^j_i \neq \hbar^j_{i'}$, then $\mu^{j}_i \neq \mu^{j}_{i'}$, for all pairwise distinct $i,i'\in \{1, \ldots ,n\}$. 
	
	As before, the base case $j = 0$ is trivial. 
	
	For the inductive step $j > 0$, let us assume that $ \hbar^j_i \neq \hbar^j_{i'} $. Then we can suppose w.l.o.g.\ that $ \hbar^{j-1}_{i,u} \notin \hbar^j_{i'}$, for some $u\in \{ 1, \ldots , m_i\}$, so that $ \hbar^{j-1}_{i,u} \neq  \hbar^{j-1}_{i',v}$, for all $ v \in \{1, \ldots , m_{i'}\} $. Hence, from the inductive hypothesis, $ \mu ^{j-1}_{i,u} \neq  \mu ^{j-1}_{i',v}$, for all $v \in \{1,\ldots,m_{i'}\}$, which implies $ \mu ^j_i \neq \mu ^j_{i'} $.
	
	\medskip
	
	Concerning (b) and recalling that all sets in a set-approximating sequence are in $ \HF $, we begin by observing that it can easily be proved, by induction on $ j $, that $ \rk(\hbar_i^j) \leq j $. Moreover, we show---again by induction on $ j $---that if $ \hbar_i^j \neq \hbar_i^{j+1} $, then $ \rk(\hbar_i^{j+1}) = j+1 $. In fact, if $ \hbar_i^0 \neq \hbar_i^{1} $, then $ \hbar_i^{1}=\{\emptyset \} $, whose rank is 1. For the inductive step, observe that if $ \hbar_i^j \neq \hbar_i^{j+1} $, then $ \hbar_{i,u}^{j-1} \neq \hbar_{i,u}^{j} $, for some $ u \in \{1, \ldots, m_i \} $. Therefore, by inductive hypothesis, $ \rk(\hbar_{i,u}^{j}) = j $, which implies $ \rk(\hbar_i^{j+1}) = j+1 $, since $\hbar_{i,u}^{j} \in \hbar_i^{j+1}$ and $\rk(\hbar_i^{j+1}) \leq j+1$.
	
	On the grounds of the above result, we can prove that if $ \hbar^j_i = \hbar^{j+1}_{i} $, then $ \hbar^{j-1}_{i,u} = \hbar^{j}_{i,u} $ for all $ u \in \{1, \ldots, m_i \} $, . In fact, assuming for contradiction that  $ \hbar^j_i = \hbar^{j+1}_{i} $ but $ \hbar^{j-1}_{i,u} \neq \hbar^{j}_{i,u} $, for some $ u \in \{1, \ldots, m_i \} $, then we would have $ \rk(\hbar^j_i)=\rk(\hbar^{j+1}_i) \geq \rk(\hbar^{j}_{i,u})+1=j+1 $, contradicting the fact that $ \rk(\hbar^j_i) \leq j $.
	
	We can now prove (b) by induction on $ j $. 
	
	For the base case $ j=0 $, if $ \hbar^0_i = \hbar^{1}_{i} $, then $\hbar^{1}_{i} = \hbar^0_i =\emptyset $, so that the equation $\x_{i} = \{\}$ must be in $\ses(\x_1, \ldots, \x_n)$, because otherwise we would have $\hbar^{1}_{i} = \{\emptyset\}$. Thus, the claim easily follows.
	
	For the inductive case $ j>0 $, let $ \hbar^j_i = \hbar^{j+1}_{i} $. Then $ \hbar^{j-1}_{i,u} = \hbar^{j}_{i,u} $, for all $ u \in \{1, \ldots, m_i \} $, so that, 
by the inductive hypothesis, $ \hbar^{j+k-1}_{i,u} = \hbar^{j-1+k+2}_{i,u} $, for all $ u \in \{1, \ldots, m_i \} $ and $ k >0 $. Thus, 
\[
\hbar^{j+k+2}_{i} = \big\{\hbar^{j-1+k+2}_{i,1},\ldots, \hbar^{j-1+k+2}_{i,m_{i}}\big\} = \big\{\hbar^{j-1}_{i,1},\ldots, \hbar^{j-1}_{i,m_{i}}\big\} = \hbar^{j}_{i}\,,
\]
proving (b).

	\medskip
	
	Concerning case (c), we first prove that the inequality $ \hbar^j_i \neq \hbar^j_{i'} $ holds for every $ j \geq n $ and pairwise distinct $  i,i'\in \{1, \ldots n\} $.
	To see this, observe that, from (a), the set of  inequalities 
	\[I^j \defAs \big\{ \langle i, i'\rangle \ | \ h^j_i\neq h^j_{i'}\big\}\] 
can only increase as $j$ grows. As a matter of fact, the growth is strict until  stabilization. In fact,  by the definition of set-system, when no new inequality is established at a given stage $j$, no new inequality can be established at any later stage $ j+1, j+2, \ldots $. Hence,   in at most $n$ steps, all the inequalities that can eventually be established must actually hold.  
	
	Let us now prove  that when $I^j$ stabilizes, all the $\hbar^j_i$'s are pairwise distinct. To see this,  consider the equivalence relation $R^{j}$ over $\{\x_{1},\ldots,\x_{n}\}$ defined as:
	\[ R^{j}(\x_i, \x_{i'}) \quad \mbox{ iff } \quad \hbar_i^{j} = \hbar_{i'}^{j}.\]
	It is easy to see that $R^j$ is a bisimulation and hence, since $\ses(\x_1, \ldots, \x_n)$ is normal, $R^{j}$ must be the identity. 
	Therefore the elements in $\langle \hbar_i^{j} : i \in \{1, \ldots n\} \rangle$ are pairwise distinct. 
		
	The fact that $ \langle \mu _i^{j} : i \in \{1, \ldots n\} \rangle = \langle \hbar_i^{j} : i \in \{1, \ldots n\} \rangle $ follows immediately by Remark~\ref{multi-set-set}.\qed
\end{proof}

Point a) in the above lemma suggests that even though set and multi-set approximating sequences will eventually ``separate'' all solutions of a set system, multi-set can introduce inequalities first. This is our first motivation for  extending the notion of  $ \R $-code to multi-sets and use such code-extension to approximate regular $ \R $-codes. The second motivation is given below in Remark \ref{cardeq}.

\begin{definition}
	Given a normal set system $ \ses(\x_1, \ldots, \x_n) $ and its multi-set approximating sequence $\big\{\langle \mu_i^{j} \sT 1 \leq i \leq n\}\rangle\big\}_{j \in \nats}$, we define the corresponding \textit{code approximating sequence} $\big\{\langle \Rs(\mu_i^{j}) \sT 1 \leq i \leq n\}\rangle\big\}_{j \in \nats}$ by recursively putting, for $i \in \{1, \ldots, n\}$ and $j \in \nats$:
\begin{equation}\label{cas}
\left\{
\begin{array}{rcl}
\Rs(\mu_i^{0}) & \defAs & 0\\
\Rs(\mu_i^{j+1}) & \defAs & \sum_{u = 1}^{m_{i}} 2^{-\Rs(\mu_{i,u}^{j})}\,.
\end{array}
\right.	
\end{equation}
We also define the corresponding \emph{code increment sequence} $\big\{\langle \delta_i^{j} \sT 1 \leq i \leq n\rangle\big\}_{j \in \nats}$ by putting, for $i \in \{1, \ldots, n\}$ and $j \in \nats$:
\begin{equation}\label{deltaDef}
\delta_{i}^{j} \defAs \Rs(\mu_{i}^{j+1}) - \Rs(\mu_{i}^{j})\,.
\end{equation}
\end{definition} 

Plainly, $\Rs(\mu_i^{j}) \geq 0$, for all $i\in \{1,\ldots,n\}$ and $j \in \nats$.

\begin{rem} \label{cardeq} \begin{sloppypar}
 Consider a normal set system $ \ses(\x_1, \ldots, \x_n) $ and its solutions $ \hbar_1, \ldots, \hbar_n $. The values $ \Rs(\mu ^1_i) $ and $ \Rs(\mu ^1_{i'}) $ are equal if and only if  $ |\hbar_i|=|\hbar_{i'}| $. The values $ \Rs(\mu ^2_i) $ and $ \Rs(\mu ^2_{i'}) $ are equal if and only if the multi-sets of the cardinalities of elements in $ \hbar_i$ and $\hbar_{i'} $ are equal. The values $ \Rs(\mu ^3_i) $ and $ \Rs(\mu ^3_{i'}) $ are equal if and only if the multi-sets of multi-sets of the cardinalities of elements of elements in $ \hbar_i$ and $\hbar_{i'} $ are equal, and so on.\qed
\end{sloppypar}
\end{rem}  

We shall make use of the following elementary property.

\begin{lemma}\label{anIneq}
For $x,y \in \reals$, if $\abs{y} \leq \abs{x}$ and $xy \leq 0$, then $\abs{2^{-y} -1} \leq \abs{2^{x} -1}$.
\end{lemma}
\begin{proof}
Let $x,y \in \reals$ be such that $\abs{y} \leq \abs{x}$ and $xy \leq 0$. Plainly, we have: 
\begin{equation}\label{ineqS}
1 \leq 2^{\abs{y}} \leq 2^{\abs{x}}.
\end{equation}
Assume first that $y \leq 0 \leq x$. Then \eqref{ineqS} yields $1 \leq 2^{-y} \leq 2^{x}$, so that $0 \leq 2^{-y} -1 \leq 2^{x} -1$, and therefore $\abs{2^{-y} -1} \leq \abs{2^{x} -1}$.

On the other hand, if $x \leq 0 \leq y$, then,  by \eqref{ineqS}, we have $1 \leq 2^{y} \leq 2^{-x}$. By taking inverses, the latter yields $2^{x} \leq 2^{-y} \leq 1$, so that $0 \leq 1 - 2^{-y} \leq 1 - 2^{x} $. Hence, $\abs{2^{-y} -1} \leq \abs{2^{x} -1}$ holds again.\qed
\end{proof}

In preparation for the proof of the existence and uniqueness of a solution to the code system associated with any normal set system, we state and prove the technical properties contained in the following lemma.

\begin{lemma}\label{technical}
Let $\ses(\x_1, \ldots, \x_n)$ be a  normal set system of the form 
\[
\left\{
	\begin{aligned}
		\x_1 & =   \{\x_{1,1}, \ldots , \x_{1,m_1}\} \\
		 &   \hspace{.2cm}\vdots \\
		\x_n & =   \{\x_{n,1}, \ldots , \x_{n,m_n}\}, \\
	\end{aligned}
\right.
\]
with index map $\idxmp$, code approximating sequence $\big\{\langle \Rs(\mu_i^{j}) \sT 1 \leq i \leq n\}\rangle\big\}_{j \in \nats}$, and code increment sequence $\big\{\langle \delta_i^{j} \sT 1 \leq i \leq n\}\rangle\big\}_{j \in \nats}$. Then, for $i \in \{1,\ldots,n\}$ and $j \in \nats$, the following facts hold:
\begin{enumerate}[label=(\roman*),leftmargin=0.8cm]
	\item\label{Final000} $\Rs(\mu_i^{j+1}) =  \delta_i^{0} + \cdots + \delta_i^{j}$,	
	
	\item\label{Final00} $\delta^{0}_i = m_{i}$,
	
	
	\item\label{Final0} $\delta_i^{j+1} =\sum_{u=1}^{m_i} 2^{-\Rs(\mu ^{j}_{i,u})} \cdot (2^{-\delta_{i,u}^{j}} - 1)$,

	\item\label{Finali} $\delta^{2j+1}_i \leq 0  \leq \delta^{2j}_i$,
		
		\item\label{Finalii} $|\delta_{i}^{j+1}| \leq |\delta_{i}^{j}|$, and
		
		\item\label{Finaliii} $\lim_{j \rightarrow \infty } \delta^{j}_i =0$.
	\end{enumerate}
\end{lemma}
\begin{proof}
Statement \ref{Final000} follows by induction from \eqref{deltaDef}, whereas  \ref{Final00} follows from \eqref{cas} and \eqref{deltaDef}.

\smallskip

Statement \ref{Final0} is easily proved by the following chain of equalities:
	\begin{align*} 
	\delta_i^{j+1} &= \Rs(\mu_{i}^{j+2}) - \Rs(\mu_{i}^{j+1}) && \text{(by \eqref{deltaDef})}\\
	&= \sum_{u = 1}^{m_i } \left(2^{-\Rs\left( \mu^{j+1}_{i,u}\right)} -  2^{-\Rs\left( \mu^{j}_{i,u}\right)}\right) && \text{(by \eqref{cas})}\\
	&= \sum_{u = 1}^{m_i } \left(2^{-\left(\Rs\left( \mu^{j}_{i,u}\right)+\delta^{j}_{i,u}\right)} -  2^{-\Rs\left( \mu^{j}_{i,u}\right)}\right) && \text{(by \eqref{deltaDef})}\\
	&= \sum_{u = 1}^{m_i } 2^{-\Rs\left( \mu^{j}_{i,u}\right)} \cdot \big(2^{-\delta^{j}_{i,u}} -  1\big)\,,
	\end{align*}
for $i \in \{1,\ldots,n\}$ and $j \in \nats$.

\smallskip

Next, statement \ref{Finali} follows by applying repeatedly \ref{Final0} and by observing that, by \ref{Final00}, $\delta_{i}^{0} \geq 0$ for every $i \in \{1,\ldots,n\}$.

\smallskip

To prove \ref{Finalii}, we proceed by induction on $j$. We need to strenghten the inductive hypothesis. Specifically, besides \ref{Finalii}, we also prove that the additional statement
\begin{enumerate}[label=(\roman*$'$), start=5,leftmargin=.8cm]
\item\label{Finalii0} $2^{-\delta^{j}_{i}} \cdot \big|2^{-\delta^{j+1}_{i}} -  1\big| \leq \big|2^{-\delta^{j}_{i}} -  1\big|$
\end{enumerate}
holds for every $i \in \{1,\ldots,n\}$.

For the base case $j=0$, we have:

\allowdisplaybreaks 
\begin{equation}\label{baseCaseFinalii}
\begin{aligned} 
	\big| \delta^{1}_i\big|&= \sum_{u = 1}^{m_i } 2^{-\Rs\left(\mu^{0}_{i,u}\right)}\cdot\big|2^{-\delta^{0}_{i,u}} -  1\big| && \text{(by \ref{Final0})}\\
	&= \sum_{u = 1}^{m_i} \big|2^{-\delta^{0}_{i,u}} -  1\big|  && \text{(by \eqref{cas})}\\
	&= \sum_{u = 1}^{m_i} \big(1 - 2^{-m_{i,u}}\big) \\
    &\leq m_{i} = \big| \delta^{0}_i\big| && \text{(by \ref{Final00}),}
\end{aligned}
\end{equation}
for every $i \in \{1,\ldots,n\}$, proving \ref{Finalii} for $j=0$ (where $m_{i,u}$ stands for $m_{\idxmp(i,u)}$).

Concerning \ref{Finalii0}, we observe that, by \ref{Finali}, the inequality 
\eqref{baseCaseFinalii} yields $0 \leq -\delta^{1}_i \leq \delta^{0}_i$. Thus, the following further inequalities hold:
\begin{gather*}
1 \leq 2^{-\delta^{1}_i} \leq 2^{\delta^{0}_i} \\
0 \leq 2^{-\delta^{1}_i} - 1 \leq 2^{\delta^{0}_i} - 1 \\
0 \leq 2^{-\delta^{0}_i} \cdot \big(2^{-\delta^{1}_i} - 1\big) \leq 1 - 2^{-\delta^{0}_i} \\
2^{-\delta^{0}_i} \cdot \big|2^{-\delta^{1}_i} - 1\big| \leq \big| 2^{-\delta^{0}_i} - 1 \big|\,,
\end{gather*}
proving also \ref{Finalii0} in the base case $j=0$.

\smallskip

For the inductive step, assume that for some $j \in \nats$ the following inequality holds for every $\ell \in \{1,\ldots,n\}$:
\begin{equation}\label{wasFinalii0}
2^{-\delta^{j}_{\ell}} \cdot \big|2^{-\delta^{j+1}_{\ell}} -  1\big| \leq \big|2^{-\delta^{j}_{\ell}} -  1\big|\,.
\end{equation}
Then we have, for every $i \in \{1,\ldots,n\}$:
\begin{align*} 
	\big| \delta^{j+2}_i\big|&= \sum_{u = 1}^{m_i } 2^{-\Rs\left(\mu^{j+1}_{i,u}\right)}\cdot\big|2^{-\delta^{j+1}_{i,u}} -  1\big| && \text{(by \ref{Final0})}\\
	&= \sum_{u = 1}^{m_i } 2^{-\left(\Rs\left( \mu^{j}_{i,u}\right) + \delta^{j}_{i,u}\right)}\cdot\big|2^{-\delta^{j+1}_{i,u}} -  1\big| && \text{(by \eqref{deltaDef})}\\
	&= \sum_{u = 1}^{m_i } 2^{-\Rs\left( \mu^{j}_{i,u}\right)}\cdot 2^{-\delta^{j}_{i,u}}\cdot\big|2^{-\delta^{j+1}_{i,u}} -  1\big|\\
	&\leq \sum_{u = 1}^{m_i } 2^{-\Rs\left( \mu^{j}_{i,u}\right)}\cdot \big|2^{-\delta^{j}_{i}} -  1\big| && \text{(by \eqref{wasFinalii0})}\\
	&= \big| \delta^{j+1}_i\big| && \text{(by \ref{Final0}),}
\end{align*}
proving \ref{Finalii} for $j+1$. In addition, from \ref{Finali} and the latter inequality $\big| \delta^{j+2}_i\big| \leq \big| \delta^{j+1}_i\big|$, Lemma~\ref{anIneq} yields $\big| 2^{-\delta^{j+2}_i} - 1\big| \leq \big| 2^{\delta^{j+1}_i} -1 \big|$, from which $2^{-\delta^{j+1}_i} \cdot \big| 2^{-\delta^{j+2}_i} - 1\big| \leq \big|2^{-\delta^{j+1}_i} -1 \big|$ follows immediately, thus proving also statement \ref{Finalii0} for $j+1$. Hence, by induction, \ref{Finalii} holds.

\medskip

Finally, concerning \ref{Finaliii}, we observe that, for every $j \geq 1$, we have:
\begin{equation}\label{limitZero}
\begin{aligned} 
	\delta^{2j+1}_i &= \sum_{u=1}^{m_i} 2^{-\Rs\left(\mu ^{2j}_{i,u}\right)} \cdot \big(2^{\delta_{i,u}^{2j}} - 1\big) && \text{(by \ref{Final0})}\\
	&= \sum_{u = 1}^{m_i } 2^{-\Rs\left( \mu^{2j-1}_{i,u}\right)} \cdot 2^{-\delta^{2j-1}_{i,u}}\cdot \big(2^{-\delta^{2j}_{i,u}} -  1\big) && \text{(by \eqref{deltaDef})}\\
	&= \sum_{u = 1}^{m_i } 2^{-\Rs\left( \mu^{2j-1}_{i,u}\right)} \cdot \big(2^{-\left(\delta^{2j}_{i,u}+\delta^{2j-1}_{i,u}\right)} -  2^{-\delta^{2j-1}_{i,u}}\big)\\
	&\leq \sum_{u = 1}^{m_i } 2^{-\Rs\left( \mu^{2j-1}_{i,u}\right)} \cdot\big(2^{-\left(\delta^{2j}_{i,u}+\delta^{2j-1}_{i,u}\right)} -  1\big),
\end{aligned}
\end{equation}
where the last inequality follows from \ref{Finali}, since $\delta^{2j-1}_{i,u} \leq 0$ and therefore $-2^{-\delta^{2j-1}_{i,u}} \leq -1$.
From \ref{Finalii}, the sequence $\{|\delta_{i}^{j}|\}_{j \in \nats}$ is convergent. Therefore, by \ref{Finali},
\[
\lim_{j \rightarrow \infty} \big(\delta_{i}^{2j} + \delta_{i}^{2j-1}\big) = \lim_{j \rightarrow \infty} \big(\absb{\delta_{i}^{2j}} - \absb{\delta_{i}^{2j-1}}\big) = 0\,,
\]
so that $\lim_{j \rightarrow \infty} \big(2^{(\delta_{i}^{2j} + \delta_{i}^{2j-1})} - 1\big)  = 0$.
Hence, from \eqref{limitZero} it follows that $\lim_{j \rightarrow \infty} \delta_{i}^{2j+1} = 0$, since $2^{-\Rs\left( \mu^{2j-1}_{i,u}\right)} \leq 1$, for all $j \geq 1$, $i \in \{1,\ldots,n\}$, and $u \in \{1,\ldots,m_{i}\}$. But then we have $\lim_{j \rightarrow \infty} \absb{\delta_{i}^{j}} = 0$, which plainly implies $\lim_{j \rightarrow \infty} \delta_{i}^{j} = 0$, proving \ref{Finaliii}, and in turn completing the proof of the lemma.\qed
\end{proof}

\begin{theorem}\label{Final}
For any given normal set system, the corresponding code system admits always a unique solution.
\end{theorem}
\begin{proof}
\begin{sloppypar}
Let $\ses(\x_1, \ldots, \x_n)$ be a normal set system of the form \eqref{setSystem}, and let $\big\{\langle \Rs(\mu_i^{j}) \sT 1 \leq i \leq n\}\rangle\big\}_{j \in \nats}$ and $\big\{\langle \delta_i^{j} \sT 1 \leq i \leq n\}\rangle\big\}_{j \in \nats}$ be its code approximating sequence and code increment sequence, respectively. Also, let $\csy(\xsi_{1},\ldots,\xsi_{n})$ be the code system
\end{sloppypar}
\begin{equation*}
\left\{
	\begin{aligned}
	\xsi_1 & =  2^{-\xsi_{1,1}} + \cdots + 2^{-\xsi_{1,m_1}} \\
		 &   \hspace{.2cm}\vdots \\
	\xsi_n & =  2^{-\xsi_{n,1}} + \cdots + 2^{-\xsi_{n,m_n}}, \\
	\end{aligned}
	\right. 
\end{equation*}
associated with $\ses(\x_1, \ldots, \x_n)$. 

We first exhibit a solution of the system $\csy(\xsi_{1},\ldots,\xsi_{n})$ and then prove its uniqueness.

\smallskip

\paragraph{\textbf{Existence:}} By Lemma~\ref{technical}\ref{Final000},\ref{Finali},\ref{Finaliii}, using the Leibniz criterion for alternating series, it follows that each of the sequences $\{\Rs(\mu_i^{j})\}_{j\in \nats}$ is convergent, for $i \in \{1,\ldots,n\}$. Let us put $\alpha_{i} \defAs \lim_{j \rightarrow \infty} \Rs(\mu_i^{j})$, for $i \in \{1,\ldots,n\}$. From \eqref{cas}, we have
\[
\Rs(\mu_i^{j+1}) = \sum_{u = 1}^{m_{i}} 2^{-\Rs(\mu_{i,u}^{j})},\quad \text{for } j \in \nats\,.
\]
Then, by taking the limit of both sides as $j$ approaches infinity, it follows that 
\[
\alpha_i =  \sum_{u = 1}^{m_{i}} 2^{-\alpha_{i,u}},
\]
for $i \in \{1,\ldots,n\}$, where $\alpha_{i,u}$ is a shorthand for $\alpha_{\idxmp(i,u)}$, with $\idxmp$ the index map of $\ses(\x_1, \ldots, \x_n)$, proving that the $n$-tuple $\langle \alpha_{1},\ldots,\alpha_{n}\rangle$ is a solution of the code system $\csy(\xsi_{1},\ldots,\xsi_{n})$.

\smallskip

\paragraph{\textbf{Uniqueness:}} \begin{sloppypar}Next we prove that the solution $\langle \alpha_{1},\ldots,\alpha_{n}\rangle$ is unique. Let $\langle \alpha_{1}',\ldots,\alpha_{n}'\rangle$ be any solution of the code system $\csy(\xsi_{1},\ldots,\xsi_{n})$. To show that $\langle \alpha_{1}',\ldots,\alpha_{n}'\rangle = \langle \alpha_{1},\ldots,\alpha_{n}\rangle$ it is enough to prove that
\end{sloppypar}
\begin{equation}\label{alphaPrime}
\Rs\big(\mu_{i}^{2j}\big) \leq \alpha_{i}' \leq \Rs\big(\mu_{i}^{2j+1}\big), \quad \text{for $j \in \nats$ and $i \in \{1,\ldots,n\}$},
\end{equation} 
holds. Indeed, from \eqref{alphaPrime}, it follows immediately
\[
\alpha_{i} = \lim_{j \rightarrow \infty} \Rs\big(\mu_{i}^{2j}\big) \leq \alpha_{i}' \leq \lim_{j \rightarrow \infty} \Rs\big(\mu_{i}^{2j+1}\big) = \alpha_{i}\,,
\]
for every $i \in \{1,\ldots,n\}$, showing that $\langle \alpha_{1}',\ldots,\alpha_{n}'\rangle = \langle \alpha_{1},\ldots,\alpha_{n}\rangle$.

We prove \eqref{alphaPrime} by induction on $j$, for all $i \in \{1,\ldots,n\}$.

For the base case $j = 0$, we observe that since $\alpha_{i}' =  \sum_{u = 1}^{m_{i}} 2^{-\alpha'_{i,u}}$ (where, as usual, $\alpha'_{i,u}$ stands for $\alpha'_{\idxmp(i,u)}$), then
\[
\Rs\big(\mu_{i}^{0}\big) = 0 \leq \alpha_{i}' \leq m_{i} = \Rs\big(\mu_{i}^{1}\big), \quad \text{for $i \in \{1,\ldots,n\}$}.
\]
For the inductive step, let us assume that 
\begin{equation}\label{inductiveHyp}
\Rs\big(\mu_{i}^{2j}\big) \leq \alpha_{i}' \leq \Rs\big(\mu_{i}^{2j+1}\big), \quad \text{for } i \in \{1,\ldots,n\},
\end{equation}
holds for some $j \in \nats$, and prove that it holds for $j+1$ as well.  From \eqref{inductiveHyp} and recalling that $\alpha_{i}' =  \sum_{u = 1}^{m_{i}} 2^{-\alpha'_{i,u}}$, the following inequalities hold, for every $i \in \{1,\ldots,n\}$ and $u \in \{1,\ldots,m_{i}\}$:
\begin{gather*}
\Rs\big(\mu_{i,u}^{2j}\big) \leq \alpha_{i,u}' \leq \Rs\big(\mu_{i,u}^{2j+1}\big)\\
2^{-\Rs\left(\mu_{i,u}^{2j+1}\right)} \leq 2^{-\alpha_{i,u}'} \leq 2^{-\Rs\left(\mu_{i,u}^{2j}\right)}\\
\sum_{u=1}^{m_{i}} 2^{-\Rs\left(\mu_{i,u}^{2j+1}\right)} \leq \sum_{u=1}^{m_{i}} 2^{-\alpha_{i,u}'} \leq \sum_{u=1}^{m_{i}} 2^{-\Rs\left(\mu_{i,u}^{2j}\right)}\\
\Rs\big(\mu_{i}^{2j+2}\big) \leq \alpha_{i}' \leq \Rs\big(\mu_{i}^{2j+1}\big).
\end{gather*}
The inequalities on the last line (for $i \in \{1,\ldots,n\}$) imply in particular that we have
\[
\Rs\big(\mu_{i,u}^{2j+2}\big) \leq \alpha_{i,u}' \leq \Rs\big(\mu_{i,u}^{2j+1}\big),
\]
for every $i \in \{1,\ldots,n\}$ and $u \in \{1,\ldots,m_{i}\}$. Hence, by repeating the very same steps as above, one can deduce also the inequalities
\[
\Rs\big(\mu_{i}^{2(j+1)}\big) = \Rs\big(\mu_{i}^{2j+2}\big) \leq \alpha_{i}' \leq \Rs\big(\mu_{i}^{2j+3}\big) = \Rs\big(\mu_{i}^{2(j+1)+1}\big) ,
\]
for $i \in \{1,\ldots,n\}$, proving that \eqref{inductiveHyp} holds for $j+1$ too. This completes the induction, and also the proof of the theorem.\qed
\end{proof}

\begin{rem}
To show that the code $\R$ is well-defined over the whole $\HHF$, we proceed as follows. Given a hereditarily finite hyperset $\hbar \in \HHF$, let $\hbar_{1},\ldots,\hbar_{n}$ be the distinct elements of the transitive closure $\TrCl{\{\hbar\}}$ of $\{\hbar\}$, where $\hbar_{1} = \hbar$.
Then we have
\begin{equation}\label{systemHypersets}
\left\{
\begin{aligned}
\hbar_{1} &= \{\hbar_{1,1},\ldots,\hbar_{1,m_{1}}\}\\
&\hspace{.2cm}\vdots\\
\hbar_{n} &= \{\hbar_{n,1},\ldots,\hbar_{n,m_{n}}\}
\end{aligned}
\right.
\end{equation}
for suitable hypersets $\hbar_{i,j} \in \{\hbar_{1},\ldots,\hbar_{n}\}$, with $i \in \{1,\ldots,n\}$ and $j \in \{1,\ldots,m_{i}\}$.

Consider the set system $\ses(\x_1, \ldots, \x_n)$
\[
\left\{
	\begin{aligned}
		\x_1 & =   \{\x_{1,1}, \ldots , \x_{1,m_1}\} \\
		 &   \hspace{.2cm}\vdots \\
		\x_n & =   \{\x_{n,1}, \ldots , \x_{n,m_n}\}, \\
	\end{aligned}
\right.
\]
associated with \eqref{systemHypersets}, where $\x_{i,j} = \x_{\ell}$ iff $\hbar_{i,j} = \hbar_{\ell}$, for all $i,\ell \in \{1,\ldots,n\}$ and $j \in \{1,\ldots,m_{i}\}$. Plainly, $\ses(\x_1, \ldots, \x_n)$ is normal and $\hbar_{1},\ldots,\hbar_{n}$ is its solution. By Theorem~\ref{Final}, let $\alpha_{1},\ldots,\alpha_{n}$ be the solution to the code system associated with $\ses(\x_1, \ldots, \x_n)$. Then $\R(\hbar_{i}) = \alpha_{i}$, for $i \in \{1,\ldots,n\}$, and, in particular, $\R(\hbar) = \R(\hbar_{1}) = \alpha_{1}$. Further, it is immediate to check that, for every normal set system $\ses'(x'_1, \ldots, x'_m)$ containing $\hbar$ in its solution, say at position $\bar k \in \{1,\ldots,m\}$, if $\alpha'_{1},\ldots,\alpha'_{m}$ is the solution to the corresponding code system, then $\alpha'_{\bar k} = \alpha_{1}$. In other words, the value $\R(\hbar)$ computed by the above procedure is independent of the normal set system used. By the arbitrariness of $\hbar$, it follows that $\R(\hbar)$ is defined for every hyperset $\hbar \in \HHF$. \qed
\end{rem}

\section{A first step towards injectivity}
As already remarked, the problem of establishing the injectivity of the map $\R$ is still open. As an initial example, we provide here a very partial result. However, the arguments used in the proof below do not seem to easily generalize even to the narrower task of proving the injectivity of $\R$ over the \emph{well-founded} hereditarily finite sets only. 

\begin{lemma}\label{duecasi}
For all $i \in \nats$, we have:
\begin{enumerate}[label=(\alph*)]
\item\label{duecasi1} $\R(h_i) \neq \R(h_{i+1})$, and
\item\label{duecasi2} $\R(h_i) \neq \R(h_{i+2})$,
\end{enumerate}
where $h_{j}$ is the $j$-th element of $\HF$ in the Ackermann ordering.
\end{lemma}
 \begin{proof}
Let $i \in \nats$. From \ref{lowLemmaiii} and \ref{lowLemmaiv} of Lemma~\ref{lowLemma} and from \eqref{2toj}, we have
\begin{equation}\label{Rhiplus1}
\R(h_{i+1}) = \R(h_{i}) - \R(h_{2^{\low(i)}-1}) + \R(h_{2^{\low(i)}})\,.
\end{equation}
If $i$ is even, then $\low(i) = 0$, and therefore
\[
\R(h_{2^{\low(i)}-1}) = 0 \neq 1 = \R(h_{2^{\low(i)}})\,.
\]
On the other hand, if $i$ is odd, then $\low(i) \neq 0$, and so, by \eqref{jOdd}:
\[
\R(h_{2^{\low(i)}}) < 1 = \R(\{h_{0}\}) \leq \R(h_{2^{\low(i)}-1})\,.
\]
In any case, we have $\R(h_{2^{\low(i)}}) - \R(h_{2^{\low(i)}-1}) \neq 0$, proving \ref{duecasi1}, by \eqref{Rhiplus1}.

\medskip

Concerning \ref{duecasi2}, we begin by putting
\[
\Delta_{j} \defAs \R(h_{2^{j}}) - \R(h_{2^{j}-1})\,,
\]
for $j \in \nats$. Let us show that $\Delta_{j} \neq -1$, for all $j \in \nats$. To begin with, for $j = 0,1,2$, we have:
\begin{align*}
\Delta_{0} &= \R(h_{1}) - \R(h_{0}) = 1 \neq -1 \\
\Delta_{1} &= \R(h_{2}) - \R(h_{1}) = \R(\{h_{1}\}) - \R(h_{1}) = 2^{-1} -1 = - \frac{1}{2} \neq -1\\
\Delta_{2} &= \R(h_{4}) - \R(h_{3}) = \R(\{h_{2}\}) - \R(\{h_{0},h_{1}\}) \\
 & \phantom{{}= \R(h_{4}) - \R(h_{3}) } = 2^{-2^{-1}} - \left(1+\frac{1}{2}\right) = \frac{1}{\sqrt{2}} - \frac{3}{2} \neq -1 \,.
\end{align*}
In addition, for $j > 2$, we have $\{h_{0},h_{1},h_{2}\} \subseteq h_{2^{j}-1}$, and therefore:
\begin{align*}
\R(h_{2^{j}-1}) \geq \R(\{h_{0},h_{1},h_{2}\}) &= 2^{-\R(h_{0})} + 2^{-\R(h_{1})} + 2^{-\R(h_{2})}\\
&= 2^{-0} + 2^{-1} + 2^{-2}\\
&= 1 + \frac{1}{2} + \frac{1}{\sqrt{2}} > 2\,. 
\end{align*}
Hence, we have $\Delta_{j} \neq -1$ also for $j > 2$, since $\R(h_{2^{j}}) = 2^{-\R(h_{j})}< 1$. But then, from \eqref{Rhiplus1}, we have, for every $i \in \nats$:
\begin{align*}
\R(h_{i+2})- R(h_{i}) &= \R(h_{i+2})- R(h_{i+1}) + \R(h_{i+1})- R(h_{i}) \\
&= \Delta_{\low(i+1)}+\Delta_{\low(i)} \neq 0\,,
\end{align*}
since either $i$ or $i+1$ is even and so either $\Delta_{\low(i)}=1$ or $\Delta_{\low(i+1)}=1$, whereas, as we proved above, $\Delta_{\low(i)}\neq -1 \neq \Delta_{\low(i+1)}$. This proves \ref{duecasi2}, completing the proof of the lemma.\qed
\end{proof}

\begin{rem}
	While the value of $\R\left( \mu^{0}_i\right)$ is $0$ for any $i\in \{1, \ldots, n\}$, the value of $\R\left( \mu^{1}_i\right)$---first approximation of $\R\left( \mu_i\right)$---is the cardinality of $\mu_i$, and the subsequent approximations oscillate within the interval $\left[0,\left|\mu_i\right|\right]$. \qed
\end{rem}

\section*{Conclusions}
By turning labels into sets, the encoding proposed in this paper can be used on a variety of structures, going from labelled graphs to Kripke models. This can be done in many different ways and in \cite{DPP04,DBLP:journals/tplp/PiazzaP04} the reader can find a rather general---albeit non optimised---technique to carry out this label elimination task. A label elimination performed to optimise  code computation (or its form) is under study.

The algorithmic side of $\R$ is also under study. A possible direction towards its usage---for example in bisimulation computation---starts from the observation that only the computation of a bounded number of digits is actually necessary  to realise all the inequalities in any given set system. 

\section*{Acknowledgements} The authors thank Eugenio Omodeo for very pleasant and fruitful conversations, and an anonymous reviewer for his/her comments.

\bibliographystyle{amsalpha}
\bibliography{QAckermann}

\appendix

\end{document}